\documentclass[groupedaddress, aps, longbibliography, reprint, prx]{revtex4-1}

\usepackage{graphicx}
\usepackage{bm}
\usepackage{bbold}
\usepackage{physics}
\usepackage{xcolor}
\usepackage{enumitem}
\usepackage{amsmath, amssymb, amsthm}
\usepackage[normalem]{ulem}
\usepackage{natbib}
\usepackage{physics}
\usepackage{comment}

\newtheorem{theorem}{Theorem}

\newtheorem{fact}{Fact}[section]
\newtheorem{lemma}{Lemma}

\newcommand{\ppfcmps}{\mathcal{E}_{\rm pPFCMPS}}
\newcommand{\rpfcmps}{\mathcal{E}_{\rm rPFCMPS}}
\newcommand{\rhmps}{\mathcal{E}_{\rm rMPS}}
\newcommand{\holo}{\mathcal{E}_{\rm holoTNS}}
\newcommand{\ppfctns}{\mathcal{E}_{\rm pPFCTNS}}
\newcommand{\rpfctns}{\mathcal{E}_{\rm rPFCTNS}}
\newcommand{\rhtns}{\mathcal{E}_{\rm rTNS}}

\usepackage[colorlinks=true]{hyperref}
\hypersetup{citecolor = blue, linkcolor=blue, urlcolor=blue}

\begin{document}

\title{Pseudoentanglement from tensor networks}

\author{Zihan Cheng}
\author{Xiaozhou Feng}
\author{Matteo Ippoliti}
\affiliation{Department of Physics, The University of Texas at Austin, Austin, TX 78712, USA}

\begin{abstract}
Pseudoentangled states are defined by their ability to hide their entanglement structure: they are indistinguishable from random states to any observer with polynomial resources, yet can have much less entanglement than random states. Existing constructions of pseudoentanglement based on phase- and/or subset-states are limited in the entanglement structures they can hide: e.g., the states may have low entanglement on a single cut, on all cuts at once, or on local cuts in one dimension. Here we introduce new constructions of pseudoentangled states based on (pseudo)random tensor networks that affords much more flexibility in the achievable entanglement structures. We illustrate our construction with the simplest example of a matrix product state, realizable as a staircase circuit of pseudorandom unitary gates, which exhibits pseudo-area-law scaling of entanglement in one dimension. We then generalize our construction to arbitrary tensor network structures that admit an isometric realization. A notable application of this result is the construction of pseudoentangled `holographic' states whose entanglement entropy obeys a Ryu-Takayanagi `minimum-cut' formula, answering a question posed in \href{https://arxiv.org/abs/2211.00747}{[Aaronson {\it et al.}, arXiv:2211.00747]}.
\end{abstract}

\maketitle

{\it Introduction.}---Entanglement is a fundamental resource in quantum information science and an important principle in theoretical physics, from condensed matter to gravity~\cite{horodecki_quantum_2009,kitaev_topological_2006,levin_detecting_2006,swingle_constructing_2012,maldacena_cool_2013,liu_entanglement_2014,kaufman_quantum_2016,abanin_colloquium_2019}. At the same time, the entanglement structure of large many-body states is very challenging to observe. Specifically, there exist $N$-qubit states with very low entanglement that are {\it computationally indistinguishable} (i.e., indistiguishable to any observer whose resources scale polynomially in $N$) from states with much higher entanglement, such as Haar-random states. This property is knowns as pseudoentanglement~\cite{aaronson_quantum_2023,giurgica-tiron_pseudorandomness_2023,jeronimo_pseudorandom_2024,bouland_public-key_2023}.

Motivated originally by ideas in quantum cryptography~\cite{ji_pseudorandom_2018}, the concept of pseudoentanglement enables connections to many-body physics in contexts including thermal equilibrium in many-body dynamics~\cite{feng_dynamics_2024}, learning of Hamiltonian ground states~\cite{bouland_public-key_2023}, and holography~\cite{aaronson_quantum_2023}, where it may have implications for the complexity of the AdS/CFT ``dictionary''~\cite{bouland_computational_2019}. In general, it is a powerful tool to demonstrate the existence of entanglement structures or other quantum features that are computationally inaccessible~\cite{gu_little_2023,gu_magic-induced_2024}, and thus contributes to sharpen the limits of what is learnable in realistic quantum experiments.

Existing constructions of pseudoentangled state ensembles are based on subset-phase-states, $\ket{\psi_{\mathcal{S},f}} \propto \sum_{x\in \mathcal{S}} (-1)^{f(x)}\ket{x}$, where under suitable constraints on the Boolean function $f:\{0,1\}^N\to \{0,1\}$ and/or the subset $\mathcal{S}\subseteq \{0,1\}^N$ one can prove computational pseudorandomness and enforce an upper bound on entanglement either on specific cuts or on all cuts at once~\cite{aaronson_quantum_2023,giurgica-tiron_pseudorandomness_2023,jeronimo_pseudorandom_2024}. 
These constructions are very powerful and in some cases achieve the theoretical maximum amount of hideable entanglement. However, they natively represent unstructured states where all $N$ qubits are on the same footing, without e.g. a notion of locality. This makes them well suited to emulate maximally random states such as those that describe infinite-temperature equilibrium~\cite{feng_dynamics_2024}, but limits their ability to mimic more structured systems. While it is possible to tailor phase-states to reflect some locality structure~\cite{bouland_public-key_2023}, it remains not clear in general how to mimic physical systems of interest, such as holographic ones.

Motivated by these considerations, in this work we introduce new families of pseudoentangled states with a highly flexible entanglement structure. We leverage tensor networks~\cite{cirac_matrix_2021,evenbly_tensor_2011,zaletel_isometric_2020,jahn_holographic_2021,haghshenas_variational_2022}, which allow us to hard-code locality and the associated entanglement (pseudo-)area-law into our states. Recently proposed efficient constructions of pseudorandom unitaries~\cite{metger_simple_2024} then allow for the efficient preparation of pseudoentangled states whose true (computationally hidden) entanglement structure is that of an arbitrary tensor network that admits an {\it isometric} realization~\cite{zaletel_isometric_2020,haghshenas_variational_2022,soejima_isometric_2020,tepaske_three-dimensional_2021}. These include many examples of interest, from matrix product states (MPS) to projected entangled pair states (PEPS) in any dimension~\cite{cirac_matrix_2021,verstraete_criticality_2006,schuch_computational_2007} and even holographic tensor network states~\cite{pastawski_holographic_2015,hayden_holographic_2016,jahn_holographic_2021,vasseur_entanglement_2019}. 
While random tensor networks have been extensively studied in various contexts~\cite{hayden_holographic_2016,brandao_efficient_2016,nahum_quantum_2017,qi_exact_2013,vasseur_entanglement_2019}, establishing pseudoentanglement requires new technical results in a distinct scaling regime for the bond dimension and the number of state copies, which we address in this work.

{\it Pseudoentanglement.}---An ensemble of pure quantum states $\mathcal{E}$ is pseudoentangled if~\cite{aaronson_quantum_2023} 
(i) it is efficiently preparable, 
(ii) it is computationally indistinguishable from the Haar-random state ensemble, and 
(iii) it has abnormally low entanglement entropy on some cuts. 
Formally, computational indistinguishability is stated as follows: for all $m \leq O({\rm poly}(N))$ and any efficient algorithm\footnote{$\mathcal{A}$ is a binary measurement on $m$ copies of the Hilbert space that can be implemented with polynomial resources.} $\mathcal{A}$, one has
\begin{equation}
    \left| \mathcal{A}(\rho^{(m)}_{\mathcal{E}}) - \mathcal{A}(\rho^{(m)}_{\rm Haar}) \right| \leq o(1/{\rm poly}(N)),
    \label{eq:pseudoentanglement_def}
\end{equation}
in terms of the `moment operators' $\rho^{(m)}_{\mathcal E} = \mathbb{E}_{\psi\sim\mathcal E}[\ketbra{\psi}^{\otimes m}]$. This means that an observer with access to polynomially many copies $m$ of a given state would have no way of deciding which ensemble it came from in polynomial time. 
The constraint on entanglement is simply that, for some extensive subsystem $A$ and with high probability over the state $\psi \sim \mathcal{E}$, the entanglement entropy $S(A)$ is bounded above by a function $f(N) \leq o(N)$. By comparison, the scaling of entanglement entropy in Haar-random states (known as the Page curve~\cite{page_average_1993}) is $S(A) = \Theta(N)$ for any extensive subsystem.
The discrepancy between true and apparent entanglement is called the pseudoentanglement gap. The theoretically maximal pseudoentanglement gap is $f(N) = \omega(\log N)$ vs $\Theta(N)$, and can be saturated by random subset-phase states~\cite{aaronson_quantum_2023,giurgica-tiron_pseudorandomness_2023,jeronimo_pseudorandom_2024,bouland_public-key_2023}.

\begin{figure}
    \centering
    \includegraphics[width=\columnwidth]{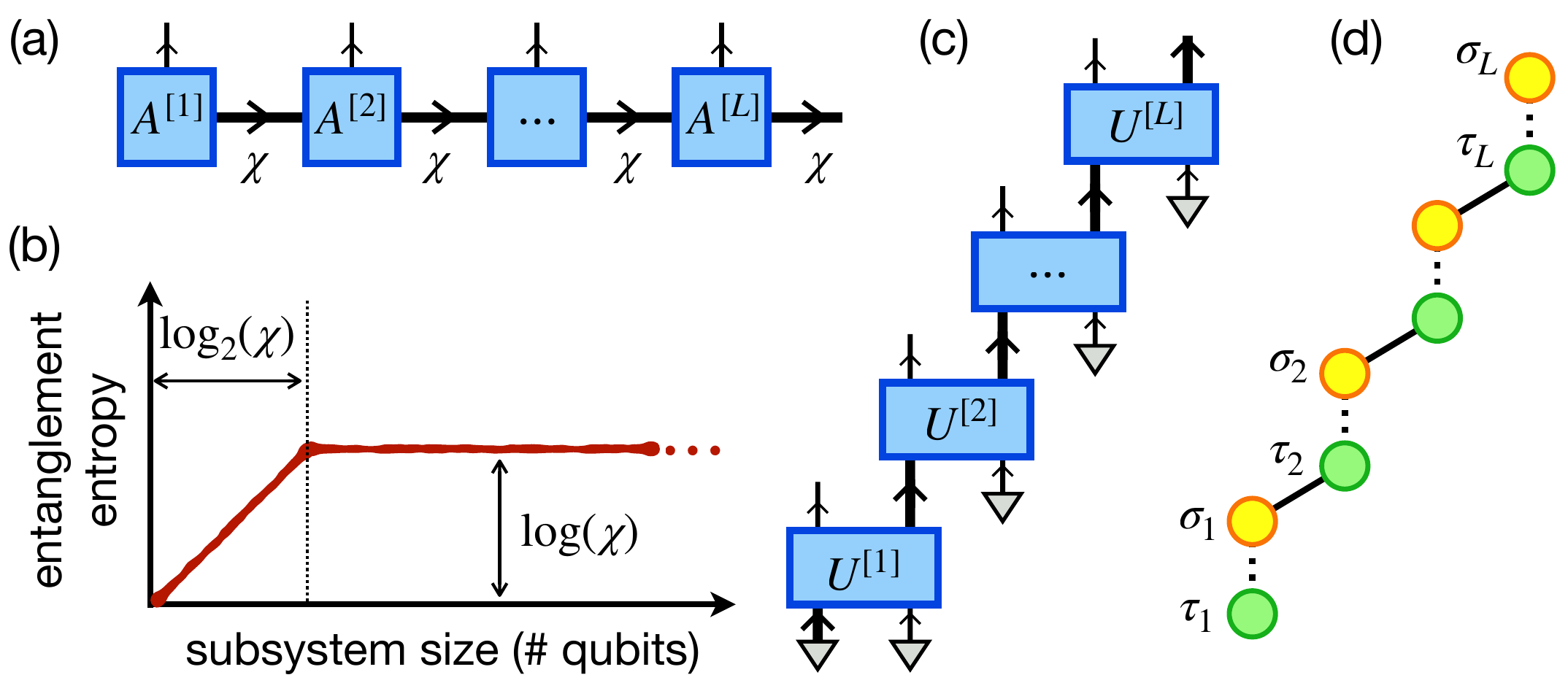}
    \caption{Construction of pseudoentangled MPS. (a) Diagrammatric representation of a MPS. Thin and thick lines represent qubits and $\chi$-state qudits, respectively; the $A^{[j]}$ tensors are isometries in the direction specified by the arrows. (b) Entanglement entropy of a typical random MPS for contiguous cuts of the qubit chain. (c) Realization of the MPS as a staircase unitary circuit. Triangles represent auxiliary qubits in the $\ket{0}$ state. (d) Partition function that controls the $m$-th moment of the random MPS ensemble. Each gate turns into two $S_m$-valued spins, with a $1D$ interaction graph.}
    \label{fig:mps}
\end{figure}

{\it Pseudoentangled matrix product states.}---We first exhibit our construction on the simplest tensor network: the one-dimensional matrix product state (MPS). A MPS, sketched in Fig.~\ref{fig:mps}(a), is a many-body state whose amplitudes are expressed as $\braket{i_1,\dots i_N}{\psi} = {\rm Tr}(A^{[1]}_{i_1} \cdots A^{[N]}_{i_N})$, with each $A^{[j]}_{i_j}$ a $\chi\times \chi$ matrix; $\chi$ is a numerical cut-off known as the {\it bond dimension}, which upper-bounds the Schmidt rank of the state about any bond on the 1D line and thus imposes an area-law, Fig.~\ref{fig:mps}(b).
A standard route to prepare MPSs in the lab is to use `staircase' circuits~\footnote{This approach is also known as ``holographic''~\cite{foss-feig_holographic_2021,anand_holographic_2023}, though not to be confused with the other usage of the word in this paper.}, where one associates $A_{i,\alpha,\beta}^{[j]} = \bra{i,\beta}U^{[j]}\ket{0,\alpha}$, with $U^{[j]}$ a unitary gate acting on a qubit and a $\chi$-state qudit. This is sketched in Fig.~\ref{fig:mps}(c). 
In this work we take $\chi = 2^\nu$, so each $U^{[j]}$ is a $(\nu+1)$-qubit gate. The staircase circuit comprises $L$ gates; for the last gate, all $\nu+1$ output qubits become part of the physical system, which in all is made of $N = L+\nu$ qubits. 
We take our unitary gates $\{ U^{[j]} \}_{j=1}^L$ to be drawn from the recently introduced `PFC ensemble'~\cite{metger_simple_2024} of pseudorandom unitaries (PRUs). 
These are the composition of a pseudorandom permutation (P), a pseudorandom binary phase function (F), and a random Clifford unitary (C)~\footnote{The sampling of random Clifford unitaries is efficient as shown in Ref.~\cite{Berg_A_2021}}, which give rise to efficiently implementable transformations that are computationally indistinguishable from Haar-random unitaries (see~\footnote{See Supplementary Information. The Supplementary Information contains Ref.~\cite{collins_weingarten_2017}.}, \cite{metger_simple_2024} for more details). 
We denote the resulting ensemble of $N$-qubit pure states by $\ppfcmps$ (pseudorandom PFC MPS).

Our first main result is that, given a sufficiently large bond dimension, this ensemble is pseudoentangled:
\begin{theorem}
If $\omega({\rm poly}(N)) \leq \chi \leq 2^{o(N)}$, the pseudorandom MPS ensemble $\ppfcmps$ is pseudoentangled.
\end{theorem}
Efficient preparability of this ensemble is clear---it is the sequential application of $O(N)$ efficient PRUs. Limited entanglement follows from the MPS structure: cutting any bond on the line yields entanglement $\leq \log(\chi)$, which is by assumption $o(N)$. It remains to prove computational indistinguishability from the Haar-random state distribution. 

We split this task in three parts. 
(I) We remove the computational derandomization, which was only needed to enable efficient preparation. By definition of the PFC ensemble~\cite{metger_simple_2024}, the pseudorandom ensemble $\ppfcmps$ is computationally indistinguishable from its random counterpart $\rpfcmps$, where the permutation and phase functions in the PFC gates are taken to be truly random instead of pseudorandom.
(II) We prove that $\rpfcmps$ is indistinguishable from its counterpart made of Haar-random unitaries, which we term $\rhmps$:
\begin{lemma} \label{lemma:pmps_vs_rmps}
    We have $\| \rho^{(m)}_{\rpfcmps} - \rho^{(m)}_{\rhmps}\|_{\rm tr} \leq O(Nm/\sqrt{\chi})$.
\end{lemma}
The proof of this fact follows easily from the properties of the PFC ensemble, and is reported in~\cite{Note4}. 
(III) We finally show that the random MPS ensemble $\rhmps$ is indistinguishable from the ensemble of global Haar-random states. This is the main technical contribution of our work:
\begin{lemma} \label{lemma:rmps_vs_haar}
    We have $\| \rho^{(m)}_{\rhmps} - \rho^{(m)}_{\rm Haar}\|_{\rm tr} \leq O(Nm^2/\chi)$.
\end{lemma}
\begin{proof}
Let us first recall two facts about the Haar measure: 
\begin{itemize}
\item[(i)] the $m$-th moment of the Haar measure over states is given by $\rho^{(m)}_{\rm Haar} = f_{q,m} \sum_{\sigma\in S_m} \hat{\sigma}$, with $\hat{\sigma}$ the replica permutation operator associated to permutation $\sigma\in S_m$, $q$ the Hilbert space dimension, and $f_{q,m} = (q-1)!/(q+m-1)!$; 
\item[(ii)] the $m$-th twirling channel for the Haar measure is given by $\Phi^{(m)}_{\rm Haar}(O) = \mathbb{E}_{U\sim {\rm Haar}} [U^{\otimes m} O (U^\dagger)^{\otimes m}] = \sum_{\sigma,\tau\in S_m} \textsf{Wg}_q(\sigma\tau^{-1}) {\rm Tr}(\hat{\tau}^\dagger O) \hat{\sigma}$, where $\textsf{Wg}_q(\sigma\tau^{-1})$ is the Weingarten function, see~\cite{Note4}.
\end{itemize}
Carrying out the Haar averages with these tools, we see that the $m$-copy distinguishability between $\rhmps$ and $\mathcal{E}_{\rm Haar}$ reads
\begin{equation}
    \Delta_m \equiv \Big\| f_{D,m} \sum_{\sigma}\hat{\sigma}^{\otimes L} - \sum_{\{\sigma_i,\tau_i\}}e^{-E[\{\sigma_i,\tau_i\}]} \bigotimes_{i=1}^L \hat{\sigma}_i \Big\|_{\rm tr} 
    \label{eq:mps_tracedist}
\end{equation}
where $D = 2^N = 2^L \chi$ is the physical Hilbert space dimension, each $\hat{\sigma}$ denotes a replica permutation for a single qubit ($i<L$) or $2\chi$-state qudit ($i=L$), and $e^{-E[\{\sigma_i,\tau_i\}]}$ is a ``Boltzmann weight'' (not necessarily positive) of a configuration of $S_m$-valued spins on a line, sketched in Fig.~\ref{fig:mps}(d). This way of writing the operator is reminiscent of a partition function~\cite{hunter-jones_unitary_2019} and will be helpful in our analysis.

Using the triangle inequality, we bound the distinguishability measure as $\Delta_m \leq \Delta_m^{\rm u.} + \Delta_m^{\rm n.u.}$, in terms of uniform and non-uniform spin configurations: 
\begin{align}
    \Delta_m^{\rm u.}
    & = \Big\| \sum_{\sigma\in S_m} \left[ f_{D,m}- e^{-E[\{\sigma,\sigma\}]}\right] \hat{\sigma}^{\otimes L} \Big\|_{\rm tr}, \label{eq:dist_unif} \\
    \Delta_m^{\rm n.u.} & = \Big\| \sum_{\{\sigma_i,\tau_i\}\text{ n.u.}} e^{-E[\{\sigma_i,\tau_i\}]} \bigotimes_{i=1}^L \hat{\sigma}_i \Big\|_{\rm tr}, \label{eq:dist_nonunif}
\end{align}
with `n.u.' denoting a restriction to non-uniform configurations (where not all $\sigma_i$, $\tau_i$ coincide).

We first show that $\Delta_m^{\rm u.}$ is small.
The Boltzmann weight of the uniform spin configuration is $e^{-E[\{\sigma,\sigma\}]} = \chi^{m(L-1)} \textsf{Wg}_{2\chi}(e)^L$ (independent of $\sigma$), with $e$ the identity permutation. We show in \cite{Note4} that $\textsf{Wg}_{2\chi}(e) = (2\chi)^{-m}[1+O(m^2/\chi)]$; therefore $e^{-E[\{\sigma,\sigma\}]} = D^{-m} [ 1+O(m^2/\chi)]$. At the same time, by the `birthday asymptotics' we have $f_{D,m} = D^{-m} [1+O(m^2/\chi)]$, so overall 
\begin{equation} \label{eq:delta_unif_small}
    \Delta_m^{\rm u.} \leq \|\rho_{\rm Haar}^{(m)} \|_{\rm tr} \left|1-\frac{e^{-E[\{e,e\}]}}{f_{D,m}}\right| \leq O \left(\frac{m^2}{\chi}\right).
\end{equation}

Next, we show that $\Delta_m^{\rm n.u.}$, Eq.~\eqref{eq:dist_nonunif}, is also small. 
We write $\bigotimes_{i=1}^L \hat{\sigma}_i = \hat{\sigma}_1^{\otimes L} \bigotimes_{i=1}^L \hat{\sigma}_i'$, where $\sigma_1' = e$ (identity permutation). Then, using the inequality $\| AB\|_{\rm tr} \leq \|A\|_{\rm tr} \|B\|_{\rm op}$ to factor out the sum over $\sigma_1$, we obtain 
\begin{align}
    \Delta_m^{\rm n.u.} 
    & \leq \|\rho_{\rm Haar}^{(m)} \|_{\rm tr} \Bigg \| \sum_{\{\sigma_i',\tau_i\}\text{ n.u.}} \frac{e^{-E[\{\sigma_i',\tau_i\}]}}{f_{D,m}} \bigotimes_{i=1}^L \hat{\sigma}_i' \Bigg \|_{\rm op}.
\end{align}
Using the triangle inequality and the fact that $\|\hat{\sigma}\|_{\rm op} = 1$ (unitarity of the replica permutations), we have
\begin{align}
        \Delta_m^{\rm n.u.} & \leq 2\sum_{\{\sigma_i',\tau_i\}\text{ n.u.}} e^{-(E^+[\{\sigma_i',\tau_i\}] - E^+[\{e,e\}])}. \label{eq:delta_nonunif_partition}
\end{align}
Here we defined positive Boltzmann weights $e^{-E^+[\{\sigma_i',\tau_i\}]} \equiv | e^{-E[\{\sigma_i',\tau_i\}]}|$, and used the fact that $f_{D,m} \simeq e^{-E^+[\{e,e\}]}\geq e^{-E^+[\{e,e\}]}/2$ (as seen in the discussion of $\Delta_m^{\rm u.}$) to subtract the ground state energy $E^+[\{e,e\}]$. 

Having pinned the $\sigma_1'$ site spin, we are free to move to bond spin variables: $\beta_\ell = \alpha_i \alpha_j^{-1}$ for each bond $\ell = (i,j)$. Each bond sum has a unit contribution from the bond ground state, $\beta_\ell = e$, while excited states $\beta_\ell \neq e$ are suppressed by powers of $1/\chi$. We show in \cite{Note4} that each such sum yields $1+O(m^2/\chi)$. Overall, $\Delta_m^{\rm n.u.} \leq 2 |[1+O(m^2/\chi)]^{2L}-1| \leq O(Nm^2/\chi)$.
\end{proof}

Eq.~\eqref{eq:delta_nonunif_partition} has an intuitive physical interpretation: the right hand side represents thermal fluctuations of a genuine lattice magnet (with real energy levels); pinning the value of $\sigma_1 = e$ removes the giant $m!$ degeneracy of the ground state, and restricting the sum to non-uniform configurations $\{\sigma_i',\tau_i\}$ further removes the (now unique) ground state $\sigma_i' = \tau_i' = e$, so that the remaining sum quantifies thermal fluctuations. Normally in 1D systems these would be large due to the lack of long range order at any finite temperature. However, the effective temperature of this magnet, controlled by the bond dimension $\chi$ as $1/\log(\chi)$, is being taken to zero as $o(1/\log(N))$ with increasing system size [since $\chi = \omega({\rm poly}(N))$]. In this regime long range order is possible even in 1D. 

This result establishes that an isometric MPS made of random unitaries is a pseudorandom state if the bond dimension $\chi$ is superpolynomial in $N$. Further, it is clear that the entanglement entropy is bounded above by $|\partial A| \log_2(\chi)$ for any subsystem $A$, where $\partial A$ is the boundary of $A$. For typical realizations, the entanglement entropy of a segment $A$ of length $\ell$ scales linearly in $\ell$ up to $\ell \simeq \nu$, at which point it saturates. This saturation value is constant in $\ell$ but grows with total system size as $\omega(\log N)$, a behavior we may call `pseudo-area-law', sketched in Fig.~\ref{fig:mps}(b).

\begin{figure}
    \centering
    \includegraphics[width=\columnwidth]{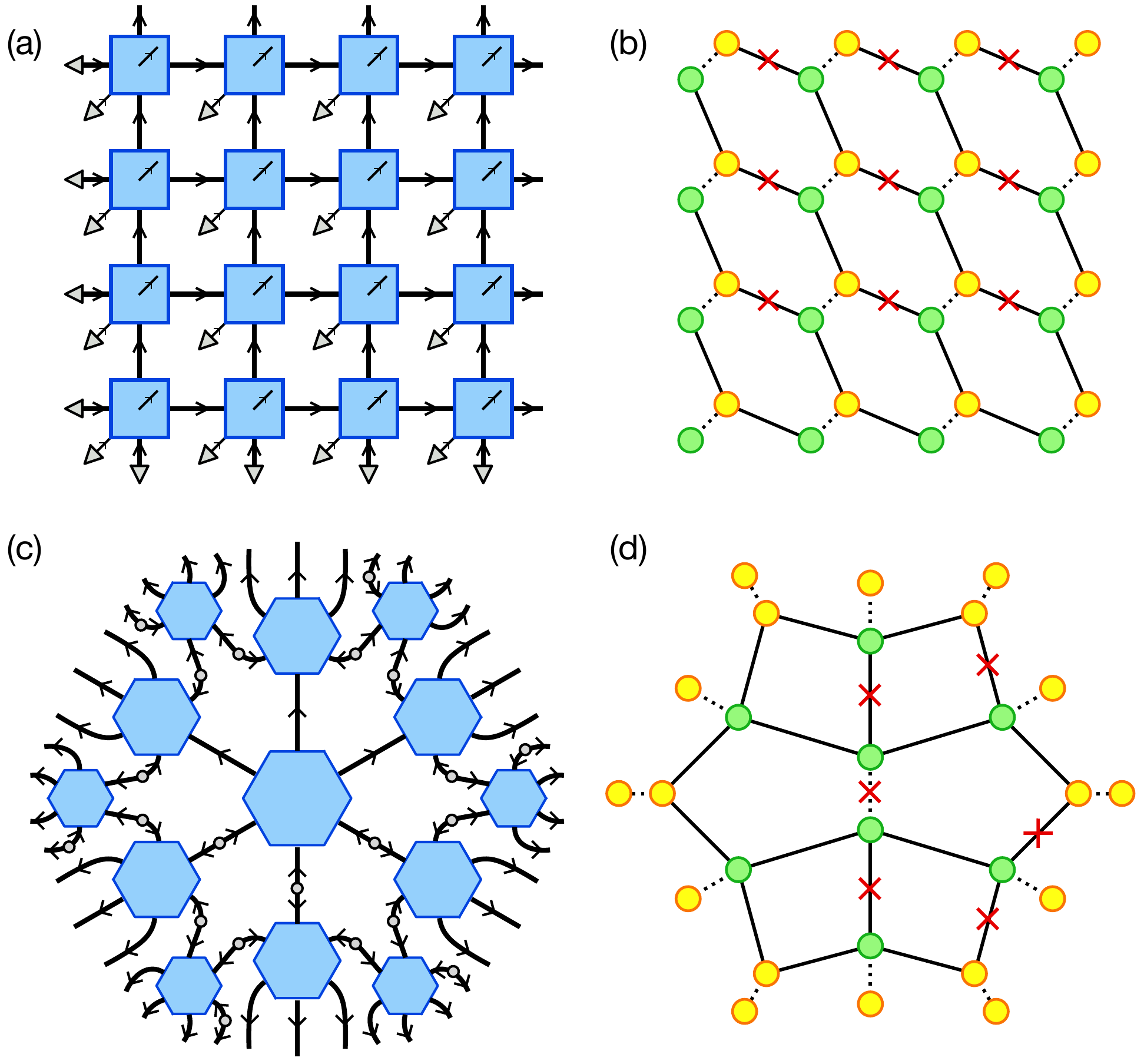}
    \caption{Construction of pseudoentangled states bsed on arbitrary isometric tensor networks. (a-c) Isometric realization of a PEPS on a $2D$ square lattice and a holographic tensor network on a $\{6,4\}$ tiling of the hyperbolic plane. Arrows denote the direction of unitarity. Gray triangles in (a) are $\ket{0}$ states, gray circles in (c) are Bell pair states. Thin lines (in (a) only) are qubits, thick lines are $\chi$-state qudits. (b-d) Partition functions obtained from averaging over Haar-random unitaries. Green and yellow sites represent ``bulk'' spins ($\tau$) and ``boundary'' spins ($\sigma$), respectively. Solid bonds give overlaps between permutations, dashed bonds give Weingarten functions. Bonds marked by a $\times$ can be dropped to obtain a spanning tree of the interaction graph, used in our proof of indistinguishability. }
    \label{fig:peps}
\end{figure}

{\it Pseudoentanglement in general isometric tensor networks.}---This construction is straightforwardly generalizable to other geometries beyond one dimension, e.g. to PEPS in 2D~\cite{verstraete_criticality_2006,cirac_matrix_2021}, which can also be realized as isometric tensor networks~\cite{zaletel_isometric_2020} subject to some constraints on their bond dimensions, see Fig.~\ref{fig:peps}(a-b). 
We generalize the proof for MPS given above to arbitray isometric tensor networks in \cite{Note4}; here we summarize the main idea.

The only element of the MPS proof that needs a nontrivial adjustment is the treatment of the partition function following Eq.~\eqref{eq:delta_nonunif_partition}, where the change of variables from site spins to bond spins exploited the loop-free $1D$ geometry of the MPS. 
However, since the Boltzmann weights $e^{-(E^+[\{\sigma_i',\tau_i\}] - E^+[\{e,e\}])}$ can be written as products of bond terms that are all individually $\leq 1$ (see \cite{Note4} for more details), the partition function is non-decreasing under removal of any given bond (i.e., replacing a bond term by $1$). 
Thus we can upper-bound the partition function by dropping edges in the graph until we obtain a spanning tree, as sketched in Fig.~\ref{fig:peps}(c-d); at this point we can change variables to bond spins (since the tree is free of loops) and conclude the proof as before. 

Therefore our construction can realize pseudo-area-law entangled states on $d$-dimensional Euclidean lattices, where the entropy of contiguous subsystems obeys $S(A) \simeq \nu |\partial A|$, with $\nu = \log_2(\chi) = \omega(\log (N))$. Such states correctly track the Page curve as long as $A$ is a ball of radius $O(\nu^{1/d})$ in the lattice, then switch to the area-law scaling.

\begin{figure}
    \centering
    \includegraphics[width=\columnwidth]{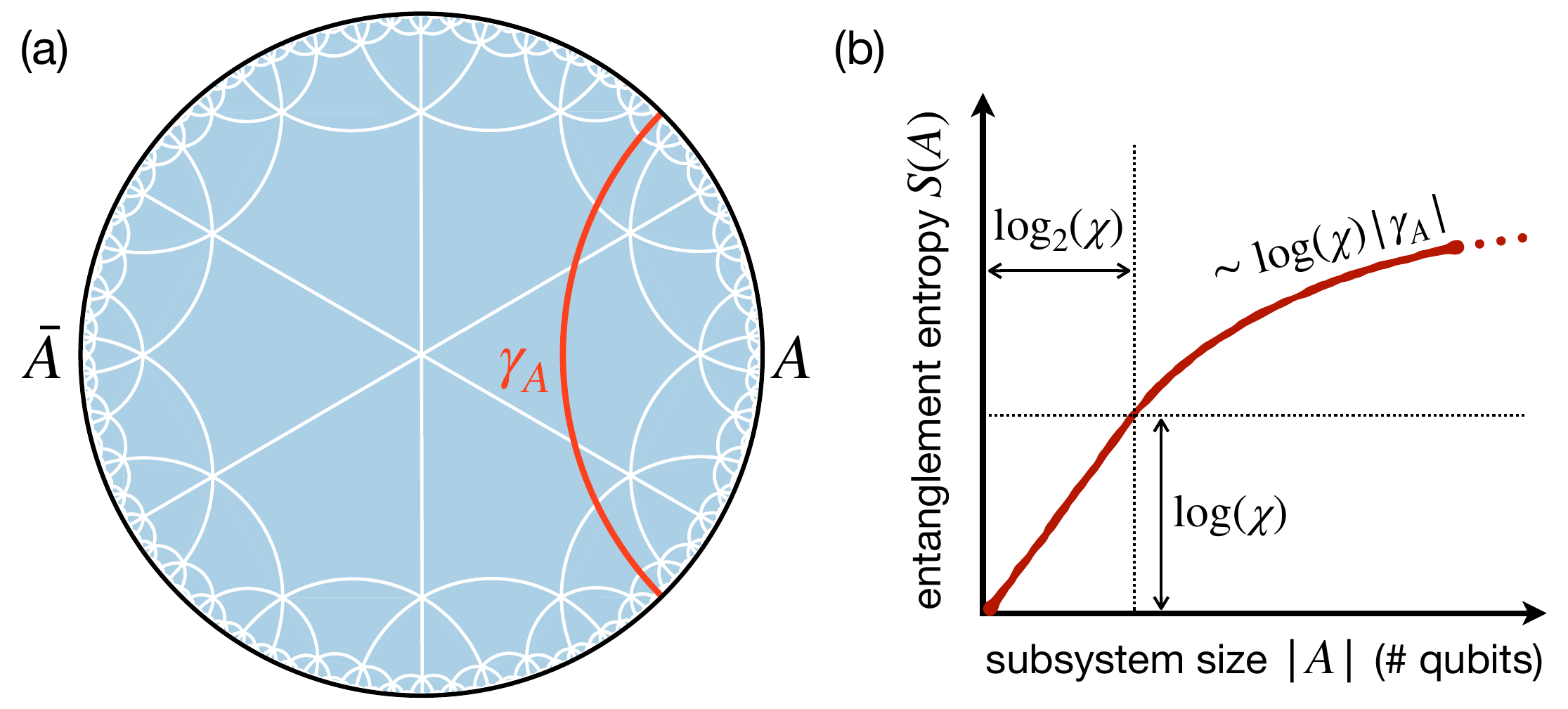}
    \caption{(a) Schematic of minimum cut for the entanglement of a boundary interval $A$ in a holographic tensor network. (b) Scaling of the entanglement entropy in holographic pseudoentangled states (here each $\chi$-state qudit in the output state is split into $\nu = \log_2(\chi)$ qubits). A Page-like scaling $\sim |A| \log(2)$ extends up to $|A| \sim \log(\chi) = \omega(\log N)$, then the minimum cut $|\gamma_A|$ starts deviating from the volume $|A|$; in hyperbolic tilings one has $|\gamma_A| \sim \log|A|$ on large length scales.}
    \label{fig:mincut}
\end{figure}

{\it Holographic tensor networks.}---An important application of our construction is to tensor network models of holography, where a quantum system is viewed as the ``boundary'' of a higher-dimensional ``bulk'' space, and its entanglement structure reflects the geometry of the bulk. 
The existence of holographic pseudoentangled states, capable of hiding the bulk geometry from observers in the boundary quantum system, may have implications for the complexity of the holographic ``dictionary'' (the mapping between observables in the bulk and boundary theories)~\cite{bouland_computational_2019,aaronson_quantum_2023}. 
Indeed, the dictionary is conjectured~\cite{susskind_complexity_2016} to relate an observable that is easy to measure (wormhole volume in the bulk gravity theory) to one that is extremely hard to measure (circuit complexity in the boundary quantum theory), seemingly violating standard complexity-theoretic assumptions. 
One proposed resolution is that the dictionary itself may be highly complex~\cite{bouland_computational_2019}. There are two stages where this complexity might arise for an observer in the boundary theory: learning the boundary entanglement structure, and using it to reconstruct the bulk geometry. A holographic pseudoentangled state would ensure hardness of the former step, potentially allowing the latter step to be efficient (as was recently argued in a toy model~\cite{aaronson_discrete_2023}). For this reason, Ref.~\cite{aaronson_quantum_2023} asked whether it is possible to construct holographic pseudoentanglement. 

Holographic tensor networks have emerged as useful toy models of the AdS/CFT correspondence and its links to quantum error correction~\cite{pastawski_holographic_2015,swingle_constructing_2012,qi_exact_2013,hayden_holographic_2016,qi_holographic_2017,jahn_holographic_2021}. They generally feature tensors arranged on a tiling of $2D$ hyperbolic space terminating with ``dangling'' legs on the boundary which represents the physical system of interest. This structure is very well suited to our construction, since it encodes the bulk isometrically into the boundary; i.e., the tensor network may be viewed as a circuit propagating information unitarily outward from the center of the hyperbolic space toward its boundary.

Here we focus for concreteness on the $\{6,4\}$ tiling of the hyperbolic plane with 6-leg tensors (3-qudit unitary gates), without bulk degrees of freedom. This is sketched in Fig.~\ref{fig:peps}(b). Placing a PRU of dimension $\chi^3$ on each hexagon defines a state ensemble $\holo$. 
Our general proof extends to this ensemble and is sufficient to show pseudoentanglement: $\holo$ is computationally indistinguishable from Haar-random, while the entanglement entropy is bounded away from volume-law. The latter property follows, as in previous examples, from the limit on the Schmidt rank of the state imposed by the tensor network structure. Namely, we have $S(A) \leq \min_{\gamma_A} |\gamma_A| \log(\chi)$, where $\gamma_A$ is a curve that separates $A$ from $\bar{A}$ in the network, and $|\gamma_A|$ is the number of bonds it cuts; see Fig.~\ref{fig:mincut}. Generally on hyperbolic tilings one has $\min_{\gamma_A}|\gamma_A| \sim \log|A|$ for contiguous $A$, and thus sub-volume-law scaling, as claimed. 

However this does not yet establish holographic entanglement, as it is only an upper bound. A key property of holographic states (and random circuits or tensor networks more generally~\cite{nahum_quantum_2017,zhou_emergent_2019,li_statistical_2021,li_entanglement_2023}) is that their entropy approximately saturates this minimum-cut upper bound. Precisely, the Ryu-Takayanagi formula~\cite{ryu_holographic_2006,ryu_aspects_2006} states that
\begin{equation}
    S(A) = \frac{\min_{\gamma_A} |\gamma_A|}{4G} \label{eq:rt_formula}
\end{equation}
where $G$ is Newton's constant in the bulk gravitational theory. 
We can prove that our ensemble $\holo$ is indeed holographic in this sense.
In \cite{Note4} we prove a lower bound on the average entropy which, together with the aforementioned upper bound form the Schmidt rank, yields the following:
\begin{theorem}\label{thm:holo_RT}
    For any subsystem $A$, the holographic pseudorandom tensor network ensemble $\holo$ obeys
    \begin{equation}
        \mathbb{E}_{\psi\sim \holo} [S(A)] = |\gamma_A| \log(\chi)[1+o(1)] + O(N/\chi). \label{eq:thm_rt}
    \end{equation}
\end{theorem}
Up to terms that vanish in the thermodynamic limit, this is a RT formula with ``Newton's constant'' $G = 4/\log(\chi) = o(1/\log(N))$. Our proof follows the standard approach for holographic random tensor networks~\cite{hayden_holographic_2016}; the reduction of our PRU-based construction to a random calculation is made possible by the fact that the PFC ensemble forms an {\it exact} 2-design (due to the random Clifford part). 

Combining Theorem~\ref{thm:holo_RT} with the proven computational indistinguishability~\cite{Note4}, we obtain the existence of holographic pseudoentangled states, answering the question posed in Ref.~\cite{aaronson_quantum_2023}. 

{\it Discussion.}---We have constructed new models of pseudoentangled quantum state ensembles by combining pseudorandom unitaries and isometric tensor networks. Compared to prior constructions of pseudoentanglement, this approach provides more flexibility in the choice of entanglement structures that may be computationally hidden. These include pseudo-area-law scaling in any dimension as well as holographic entanglement which obeys a Ryu-Takayanagi formula, Eq.~\eqref{eq:thm_rt}. It would be interesting to extend our construction to more general and non-isometric holographic codes~\cite{akers_black_2022,dewolfe_non-isometric_2023}. More broadly, the implications of our construction for holography remain to be explored; it is worth emphasizing that tensor networks are a very simplistic toy model of holography and further work is needed to extend these results toward more realistic models. 

Given the prominent role of tensor networks in many-body physics (from exact representations of local Hamiltonian ground states to variational ans\"{a}tze for time evolution to models of holography) our results facilitate the integration of computational pseudorandomness in diverse physical settings. This paves the way for applications to the hardness of learning various properties of structured many-body states, such as Hamiltonian eigenstates or states prepared by specific types of non-thermalizing quantum dynamics. We leave these interesting directions to future work. 

{\it Note added.}---While working on this manuscript we became aware of related work by Akers, Bouland, Chen, Kohler, Metger and Vazirani~\cite{akers_holographicpseudoentanglement_2024} that introduces a different construction of holographic pseudoentangled states.
Additionally, during or after completion of this manuscript, several related works appeared on arXiv: Ref.~\cite{engelhardt_spoofing_2024} by Engelhardt et al introduced a different construction of pseudoentangled holographic states; 
Ref.~\cite{schuster_random_2024} by Schuster, Haferkamp and Huang introduced a construction of pseudoentangled states based on shallow circuits of PFC PRUs---their 1D construction is similar, but not identical, to our MPS construction, and the techniques used in our proofs are different;
finally Ref.~\cite{lami_anticoncentration_2024} by Lami, De Nardis and Turkeshi analitically and numerically studied random MPS and PEPS states in the regime of $\chi = {\rm poly}(N)$, finding results that become consistent with ours when extrapolated to $\chi = \omega({\rm poly}(N))$.

\begin{acknowledgments}
\textit{Acknowledgments.}---MI gratefully acknowledges discussions with Adam Bouland, Tudor Giurgica-Tiron, and Nick Hunter-Jones.
\end{acknowledgments}

\let\oldaddcontentsline\addcontentsline
\renewcommand{\addcontentsline}[3]{}
\bibliography{pseudoentangled_rtns}
\let\addcontentsline\oldaddcontentsline

\clearpage
\widetext

\setcounter{equation}{0}
\setcounter{figure}{0}
\setcounter{table}{0}
\setcounter{page}{1}
\makeatletter
\renewcommand{\thesection}{S\arabic{section}}
\renewcommand{\theequation}{S\arabic{equation}}
\renewcommand{\thefigure}{S\arabic{figure}}

\begin{center}
\textbf{\large Supplementary Information: Pseudoentanglement from tensor networks} \\~\\
Zihan Cheng, Xiaozhou Feng, and Matteo Ippoliti \\
\textit{Department of Physics, The University of Texas at Austin, Austin, TX 78712, USA}
\end{center}

\tableofcontents

\section{Bounds on Weingarten functions \label{sec:si_wg}}

Here we prove several facts about the Weingarten function which are used in the main text toward proving Theorem 1. 
To this end, it is helpful to first recall the relationship between Weingarten functions and the Gram matrix of permutations, $G_{\sigma,\tau} \equiv {\rm Tr}(\hat{\sigma} \hat{\tau}^{-1}) = d^{|\sigma\tau^{-1}|}$. Here $d$ is the dimension of the underlying Hilbert space, $\sigma,\tau\in S_m$ are permutations, $m$ is the number of copies (or replicas) of the Hilbert space, and $|\sigma|\in \{1,\dots m\}$ denotes the number of cycles in permutation $\sigma$. 
We have $\textsf{Wg}_d(\sigma\tau^{-1}) = (G^{-1})_{\sigma,\tau}$, with $G^{-1}$ the matrix inverse of $G$, viewed as a $m!\times m!$ square matrix. The matrix is invertible as long as $m < d$, which is the regime of interest to us. 

\begin{fact} \label{fact:G_sum}
    We have $\sum_{\beta\in S_m} d^{|\beta|} = d^m [1+O(m^2/d)].$
\end{fact}
\begin{proof}
    Noting that $d^{|\beta|} = {\rm Tr}(\hat{\beta})$, where $\hat{\beta}$ is the replica permutation operator associated to $\beta\in S_m$, we have 
    \begin{equation}
        \sum_{\beta\in S_m} d^{|\beta|}
        = {\rm Tr} \left(\sum_{\beta\in S_m} \hat{\beta} \right)
        = m!\, {\rm Tr}(\hat{\Pi}_{\rm sym}) = \frac{(d-1+m)!}{(d-1)!},
    \end{equation}
    where $\hat{\Pi}_{\rm sym} = (1/m!)\sum_\beta \hat{\beta}$ is the projector on the symmetric sector of the $m$-copy Hilbert space, which has dimension $\binom{d-1+m}{m}$. 
    The statement follows from the `birthday statistics': $\log[(d-1+m)! / (d-1)!] = m \log(d-1) + \binom{m}{2}\frac{1}{d-1} + o(1/d)$.
\end{proof}

\begin{fact} \label{fact:wg_sum_signed}
    We have $\sum_{\beta\in S_m} d^m \textsf{Wg}_d (\beta) = 1+O(m^2/d)$.
\end{fact}
\begin{proof}
    We have that $1 = {\rm Tr} \mathbb{E}_{U\sim {\rm Haar}}[U^{\otimes m}\ketbra{0} (U^\dagger)^{\otimes m}] = \sum_{\sigma,\tau\in S_m} \textsf{Wg}_d(\tau\sigma^{-1}) d^{|\sigma|}$. Changing summation variable to $\tau\mapsto \tau' = \tau \sigma^{-1}$ yields
    \begin{equation}
        \sum_{\tau' \in S_m} {\sf Wg}_d(\tau') = \left(\sum_{\sigma\in S_m} d^{|\sigma|}\right)^{-1} = \frac{(d-1)!}{(d-1+m)!} = d^{-m}[1+O(m^2/d)]
    \end{equation}
    by Fact~\ref{fact:G_sum}.
\end{proof}

\begin{fact} \label{fact:wg_abs}
    For all permutations $\sigma\in S_m$, we have 
    \begin{equation} 
    |\textsf{Wg}_d (\sigma)| = d^{|\sigma|-2m} f(\sigma) [1+O(m^2/d)],
    \end{equation}
    with $f(\sigma) = \prod_{i=1}^{|\sigma|} \frac{1}{\ell_i}\binom{2(\ell_i-1)}{\ell_i-1}$, and $\{\ell_i\}_{i=1}^{|\sigma|}$ the lengths of all cycles in $\sigma$. 
    In particular $\textsf{Wg}_d(e) = d^{-m} [1+O(m^2/d)]$.
\end{fact}
\begin{proof}
    See Ref.~\cite{collins_weingarten_2017}. Note that our usage of the notation for $|\sigma|$ differs from theirs. The second statement follows from the fact that $e$ has $m$ cycles of length 1, so $f(e) = 1$. 
\end{proof}

\begin{fact} \label{fact:wg_sum_unsigned}
    We have  $\sum_{\beta\in S_m} d^m |\textsf{Wg}_d(\beta)| = 1 + O(m^2/d)$.
\end{fact}
\begin{proof}
    First we have a lower bound $\sum_{\beta\in S_m} d^m |\textsf{Wg}_d(\beta)| \geq \sum_{\beta\in S_m} d^m \textsf{Wg}_d(\beta) = 1+O(m^2/d)$, by Fact~\ref{fact:wg_sum_signed}. 
    Second, we can prove an upper bound by invoking Fact~\ref{fact:wg_abs}:
    \begin{align}
        \sum_{\beta\in S_m} d^m |\textsf{Wg}_d(\beta)| 
        & \leq [1+O(m^2/d)] \sum_{\beta\in S_m} d^{|\beta|-m} f(\beta) \nonumber \\
        & \leq [1+O(m^2/d)] \sum_{\beta\in S_m} (d/4)^{|\beta|-m} \nonumber \\
        & \leq 1+O(m^2/d).
    \end{align}
    In the second inequality we used the fact that $\frac{1}{\ell} \binom{2(\ell-1)}{\ell-1} \leq 4^{\ell-1}$, so that $f(\beta) \leq 4^{\sum_i(\ell_i-1)} = 4^{m-|\beta|}$. We then used Fact~\ref{fact:G_sum} to carry out the sum over $\beta$. 
\end{proof}

\section{Computational indistinguishability between $\rpfcmps$ and $\rhmps$ \label{sec:si_pmps_rmps}}

Here we give a proof of Lemma~1, stating the indistinguishability between the ensemble of MPS made out of Haar-random gates, $\rhmps$, and the analogous ensemble where the MPS is made out of random `PFC' gates, $\rpfcmps$. 
To this end, we use a key property of the PFC ensemble: 
\begin{fact}\label{fact:PFC}
    Given an arbitrary state $\rho$ on $\mathcal{H}_A^{\otimes m} \otimes \mathcal{H}_B$, with $\mathcal{H}_A$ a Hilbert space of dimension $d$ and $\mathcal{H}_B$ an arbitrary Hilbert space, we have
    \begin{equation}
        \| \Phi^{(m)}_{\rm Haar}\otimes \mathcal{I} (\rho) - \Phi^{(m)}_{\rm rPFC} \otimes \mathcal{I} (\rho) \|_{\rm tr} \leq O(m/\sqrt{d}), \label{eq:PFC_property}
    \end{equation}
    where $\Phi^{(m)}_{\rm Haar}\otimes \mathcal{I} (\rho) = \mathbb{E}_{U\sim {\rm Haar}}[(U^{\otimes m}\otimes I)\rho (U^{\otimes m} \otimes I)^\dagger]$ is the $m$-fold twirling channel for the Haar ensemble (acting only on the $\mathcal{H}_A^{\otimes m}$ factor of the Hilbert space), and $\Phi^{(m)}_{\rm rPFC}$ is the same for the random PFC ensemble. 
\end{fact}
\begin{proof}
    See Ref.~\cite{metger_simple_2024}, Theorem 5.2. Note that the theorem is proven for pure states; we can obtain this statement by expanding the auxiliary Hilbert space $\mathcal{H}_B$ to purify $\rho$, and using monotonicity of the trace distance under the partial trace.
\end{proof}

The moment operators for the two ensembles $\rpfcmps$, $\rhmps$ can be written as
\begin{align}
    \rho^{(m)}_{\rhmps} 
    & = \Phi^{(m)}_{\rm Haar}|_{[L,L+\nu]} \circ \Phi^{(m)}_{\rm Haar}|_{[L-1,L+\nu-1]} \circ \cdots \circ \Phi^{(m)}_{\rm Haar}|_{[1,\nu+1]} (\rho_0), \\ 
    \rho^{(m)}_{\rpfcmps} 
    & = \Phi^{(m)}_{\rm rPFC}|_{[L,L+\nu]} \circ \Phi^{(m)}_{\rm rPFC}|_{[L-1,L-1+\nu]} \circ \cdots \circ \Phi^{(m)}_{\rm rPFC}|_{[1,\nu+1]} (\rho_0),
\end{align}
where $\rho_0 = (\ketbra{0})^{\otimes mN}$ is the $m$-fold replicated input state and $\Phi^{(m)}|_{[i,j]}$ denotes an $m$-fold twirling channel acting only on qubits $i$ through $j$ on the chain (the identity channel on all other qubits is implicit). In our staircase circuits, gates act on $\nu+1$ consecutive qubits and move over by 1 qubit at each time step. 

We aim to prove that 
\begin{equation}
    \|\rho^{(m)}_{\rhmps} - \rho^{(m)}_{\rpfcmps} \|_{\rm tr} \leq L\epsilon,
    \label{eq:si_upperbound_Leps}
\end{equation}
with $\epsilon$ the approximation error of the random PFC ensemble on $\nu+1$ qubits: $\epsilon \leq O(m/\sqrt{\chi})$, where $\chi = 2^\nu$ is the MPS bond dimension. 
We will prove the statement by induction.

\begin{itemize}
    \item ($L=1$): The whole system consists of $\nu+1$ qubits. We have $\|\Phi^{(m)}_{\rm Haar} (\rho_0)- \Phi^{(m)}_{\rm rPFC} (\rho_0) \|_{\rm tr} \leq \epsilon$. 
    \item ($L-1\implies L$): Let 
    \begin{align} 
    \rho_1 & = \Phi^{(m)}_{\rm Haar}|_{[L-1,L-1+\nu]} \circ \cdots \circ \Phi^{(m)}_{\rm Haar}|_{[1,\nu+1]} (\rho_0), \\ 
    \rho_2 & = \Phi^{(m)}_{\rm rPFC}|_{[L-1,L-1+\nu]} \circ \cdots \circ \Phi^{(m)}_{\rm rPFC}|_{[1,\nu+1]} (\rho_0).
    \end{align}
    We have, by triangle inequality,
    \begin{align}
        \| \Phi^{(m)}_{\rm Haar}|_{[L,L+\nu]}(\rho_1) - \Phi^{(m)}_{\rm rPFC}|_{[L,L+\nu]}(\rho_2) \|_{\rm tr} 
        & \leq \| \Phi^{(m)}_{\rm Haar}|_{[L,L+\nu]}(\rho_1) - \Phi^{(m)}_{\rm rPFC}|_{[L,L+\nu]}(\rho_1) \|_{\rm tr} \nonumber \\
        & \qquad + \| \Phi^{(m)}_{\rm rPFC}|_{[L,L+\nu]}(\rho_1) - \Phi^{(m)}_{\rm rPFC}|_{[L,L+\nu]}(\rho_2) \|_{\rm tr}.
    \end{align}
    The first term on the rhs is $\leq \epsilon$ by direct application of Eq.~\eqref{eq:PFC_property}. The second term, by monotonicity of the trace distance under quantum channels, is $\| \Phi^{(m)}_{\rm rPFC}|_{[L,L+\nu]}(\rho_1-\rho_2)\|_{\rm tr} \leq \| \rho_1-\rho_2 \|_{\rm tr}$; but (up to tensoring with a trivial state $\ketbra{0}^{\otimes m}$) $\rho_1-\rho_2$ is the difference between the two moment operators at step $L-1$, which by hypothesis is $\leq (L-1)\epsilon$. The result follows. 
\end{itemize}

We remark that this proof did not use the MPS structure specifically, but only the isometric structure of the tensor network (i.e. the ability to write it as a sequential application of unitary gates). As long as the network consists of ${\rm poly}(N)$ gates, the error incurred by replacing Haar-random gates by random PFC gates remains small, $O({\rm poly}(N) m / \sqrt{\chi})$, see Sec.~\ref{sec:si_gen_lemma1}

\section{Proof of pseudorandomness for arbitrary isometric tensor network geometries \label{sec:si_general_isotns}}

Here we modify the proof presented in the main text for MPS to general isometric tensor network architectures. We define the state ensembles $\ppfctns$ (made of pseudorandom PFC gates), $\rpfctns$ (made of random PFC gates) and $\rhtns$ (made of Haar-random gates) in the same ways as $\ppfcmps$, $\rpfcmps$ and $\rhmps$, just replacing the MPS architecture with a general tensor network state (TNS). 

\subsection{Setup} 

Formally, we consider a weighted directed graph whose edges are qudit worldlines (with the direction of each edge going from past to future and the weight indicating the bond dimension $\chi_i$, allowed to vary across edges) and whose vertices can be of three types: unitary, input, or output.
\begin{enumerate}
    \item[(i)] Unitary vertices must have both incoming and outgoing edges, and must respect the condition $\prod_{i\in {\rm incoming}} \chi_i = \prod_{i \in {\rm outgoing}} \chi_i$ (required by unitarity). 
    \item[(ii)] Input vertices have only outgoing edges. They represent the injection of auxiliary qudits in the system. We restrict for simplicity to product states $\ket{0}$, which have one outgoing edge, and Bell pairs $\ket{\psi_{\rm Bell}} = (1/\sqrt{\chi}) \sum_{j=0}^\chi \ket{jj}$, which have two outgoing edges of the same bond dimension $\chi$. Both types of input vertices are illustrated in Fig.~2(a,c). 
    \item[(iii)] Output vertices have only one incoming edge and no outgoing edges. They represent degrees of freedom of the final tensor network state, or ``dangling legs'' in standard tensor network notation. 
\end{enumerate}
In analogy with the MPS case, we define the number of output vertices as $L$. 
We also define $N$ as the effective system size (i.e. number of qubits) via $N = \log_2(D)$, with $D = \prod_{\ell \in {\rm output}} \chi_\ell$ the total Hilbert space dimension. $N$ need not be an integer.
We assume that the total number of unitary vertices in the graph is $n_U\leq O({\rm poly}(N))$, which ensures the states' efficient preparability. 

We aim to prove that the states obtained from these graphs by placing a PFC gate of the appropriate dimension on each unitary vertex are computationally indistinguishable from Haar-random states of dimension $D = \prod_{\ell \in {\rm output}} \chi_\ell $, where the product runs over edges connected to output vertices. (Each output vertex represents a local Hilbert space and the output state lives in the tensor product of these spaces.)

The proof is analogous to the one reported in the main text for matrix product states, with some small modifications that we address in the following.

\subsection{Generalization of Lemma~1 \label{sec:si_gen_lemma1}}

We can repeat the same argument as in Sec.~\ref{sec:si_pmps_rmps}, but instead of $O(L\epsilon)$ as the upper bound on the trace distance [Eq.~\eqref{eq:si_upperbound_Leps}]
we obtain $O(\sum_{u\in {\rm gates}} \epsilon_u)$, since in this case each gate may act on a Hilbert space of different dimension. We have $\epsilon_u \leq O(m^2/d_u)$, with $d_u = \prod_{\ell \in {\rm outgoing}(u)} \chi_\ell $, and thus 
\begin{equation}
    \| \rho^{(m)}_{ \rpfctns} - \rho^{(m)}_{ \rhtns } \|_{\rm tr} \leq O(n_U m/\sqrt{d_{\rm min}}), 
    \qquad d_{\rm min} = \min_u \prod_{\ell \in {\rm outgoing}(u)} \chi_\ell . \label{eq:si_pvsr}
\end{equation}
Here $n_U$ is the total number of unitary gates in the network, again assumed to be $O({\rm poly}(N))$. If all the $\chi_i$ are super-polynomial in $N$, then $d_{\rm min} = \omega({\rm poly}(N))$ and the two ensembles are indistinguishable in polynomial time.

\subsection{Generalization of Lemma~2}

Here we generalize Lemma~2 to the more general class of tensor network states (TNS) beyond MPSs. We aim to prove that the ensemble of TNSs built out of Haar-random gates is indistinguishable from the Haar-random ensemble on the whole Hilbert space.
In analogy with the MPS case, we can compute the $m$-copy distinguishability, perform the Haar averages, and obtain a modified graph corresponding to the partition functions in Fig.~2(b,d) in the main text. 
The new graph is undirected, weighted and labeled---each edge has one of two labels, that we denote by `solid' and `dashed' in keeping with Fig.~1 and 2 in the main text. The new graph (interaction graph of the partition function) is built from the old one (the tensor network diagram) according to these rules:
\begin{itemize}
    \item All edges between two unitary vertices become `solid' and keep their weight $\chi_i$.
    \item Output vertices and their incoming edges are removed.
    \item Input vertices are removed. For input vertices of product-state type, the outgoing edge is also removed. For input vertices of Bell-pair type, the outgoing edges are merged into a single `solid' edge whose weight is still $\chi$. 
    \item Each unitary vertex is replaced by two vertices. One is connected to all the formerly outgoing edges, the other to all the formerly incoming edges. The two vertices are connected to each other by a new `dashed' edge of weight $d_u$ (Hilbert space dimension on which the unitary acts). 
\end{itemize}
This graph serves as the interaction graph for an $S_m$-valued magnet's partition function: each vertex hosts a spin $\alpha\in S_m$, 
each `solid' edge connecting two spins $\alpha_{i,j}$ contributes a Boltzmann weight $\chi_\ell^{|\alpha_i\alpha_j^{-1}|}$, 
each `dashed' edge connecting two spins $\alpha_{i,j}$ contributes $\textsf{Wg}_{d_u}(\alpha_i\alpha_j^{-1})$, with $\chi_\ell$ and $d_u$ the edge weights as defined above. 
In addition, there is an overall prefactor of $\prod_{\ell\in {\rm Bell}} \chi_\ell^{-m}$, with the product running over Bell pair input states in the original graph, coming from the normalization prefactor in $\ketbra{\psi_{\rm Bell}}^{\otimes m}$.

The proof of computational indistinguishability proceeds in the same way as for the MPS case until the derivation of Eqs.~(3-4) in the main text. 

\subsubsection{Uniform spin configurations}

To prove that $\Delta_m^{\rm u.}$ is small, Eq.~(5) of the main text, we need to show that the Boltzmann weight $e^{-E[\{e,e\}]}$ of a uniform spin configuration is close to $D^{-m}$. 
We have 
\begin{align}
    e^{-E[\{e,e\}]} 
    & = \prod_{\ell\in{\rm Bell}} \chi_\ell^{-m} \prod_{\ell\in{\rm solid}} \chi_\ell^m \prod_{u\in{\rm dashed}} \textsf{Wg}_{d_u}(e) \nonumber \\
    & = \prod_{\substack{\ell\in\text{solid},\\ \text{not Bell}}} \chi_\ell^m \prod_{u\in{\rm dashed}} d_u^{-m} [1+O(m^2/d_u)], \label{eq:si_uniformconfig}
\end{align}
where we used the results of Sec.~\ref{sec:si_wg} to approximate the Weingarten functions, and restricted the product over edges to edges that did not originate from Bell pairs. 
Let us now think about this product in terms of the original (tensor network) graph. Writing $d_u = \prod_{\ell \in {\rm outgoing}(u)} \chi_\ell$ we may view Eq.~\eqref{eq:si_uniformconfig} as a product over edges of the graph. Edges between input vertices (whether product or Bell-pair) and unitaries do not appear in the product; edges connecting two unitaries cancel (they contibute a $\chi_\ell^m$ to the first factor and a $\chi_\ell^{-m}$ to the second);
edges between unitaries and output vertices contribute a $\chi_\ell^{-m}$ factor.
In all, the result is $(\prod_{\ell \in{\rm output}} \chi_\ell)^{-m} = D^{-m}$, with $D$ the physical Hilbert space dimension. The error terms stem only from the Weingarten functions and can be bounded above by $|[1+O(m^2/d_u)]^{n_U}-1|$ with $n_U$ the number of unitaries. 
Overall this gives 
\begin{equation}
    \Delta_m^{\rm u.} \leq O(n_U m^2 / d_{\rm min}),
\end{equation}
which is smaller than the error incurred from the replacing pseudorandom with Haar-random gates, Eq.~\eqref{eq:si_pvsr}.

\subsubsection{Non-uniform spin configurations}

The derivation of Eq.~(7) in the main text makes no reference to the MPS structure and thus proceeds unchanged. Once we have a partition function with positive Boltzmann weights
\begin{equation}
    e^{-( E^+[\{\sigma_i',\tau_i\}] - E^+[\{e,e\}])} = \prod_{(i,j) \in {\rm edges}} \left\{ 
    \begin{aligned}
        & 1& \text{ if } \alpha_i = \alpha_j \\
        & O(1/\chi_{ij}) & \text{ otherwise}
    \end{aligned}\right.,
\end{equation}
we can bound it above by simply dropping edges, i.e., setting the associated factor to 1 regardless of the spins $\alpha_{i,j}$. 
We drop edges until we obtain a spanning tree $\mathcal{T}$ of the graph, as sketched in Fig.~2(b,d), and move to bond variables $\beta_{ij} \equiv \alpha_i \alpha_i^{-1}$ starting from the pinned site $\sigma_1' = e$ and covering all links in the spanning tree.
This gives
\begin{align}
    \sum_{\{\sigma_i',\tau_i\}\text{ n.u.}} e^{-( E^+[\{\sigma_i',\tau_i\}]-E^+[\{e,e\}]) }
    & \leq -1+ \prod_{\ell \in \mathcal{T}, \text{ solid}} \left( \sum_\beta \chi_\ell^{|\beta|-m} \right) 
    \prod_{u \in \mathcal{T}, \text{ dashed}} \left( \sum_\beta \frac{|\textsf{Wg}_{d_u}(\beta)|}{\textsf{Wg}_{d_u}(e)} \right) \nonumber \\
    & \leq -1 + \prod_{\ell \in \mathcal{T}, \text{ solid}} [1+O(m^2/\chi_\ell)]
    \prod_{u \in \mathcal{T}, \text{ dashed}} [1+O(m^2/d_u)] \nonumber \\
    & \leq O(n_E m^2/\chi_{\rm min}) + O(n_U m^2/d_{\rm min}).
\end{align}
Here $\chi_{\rm min}$ is the minimum bond dimension of an inner bond (i.e., a bond connecting unitary gates---physical bonds are excluded), and $n_E$ is the number of inner bonds. We have that $d_{\rm min} \geq \chi_{\rm min}$ and that $n_E = \Theta(n_U)$, so we can simplify the error term as $\leq O(n_U m^2/\chi_{\rm min})$. This may be larger or smaller than the error incurred from replacing pseudorandom gates with Haar-random gates, Eq.~\eqref{eq:si_pvsr}, depending on the connectivity of the circuit, i.e. whether $\chi_{\rm min} < \sqrt{d_{\rm min}}$ or not.

In summary, we conclude that 
\begin{equation}
    \| \rho^{(m)}_{ \rhtns } - \rho^{(m)}_{\mathcal{E}_{\rm Haar}} \|_{\rm tr} \leq O(n_U m^2 / \chi_{\rm min}) + O(n_U m /\sqrt{d_{\rm min}}).
\end{equation}
For the case of uniform inner bond dimension $\chi$ across the network and $n_U = O(N)$, this yields the MPS result: $O(Nm^2/\sqrt{\chi})$. However this shows that the result is much more general, and that arbitrary isometric tensor network geometries still yield computationally random states as long as $n_U$ is not too large and there are no ``weak links'' where $\chi$ or $d_u$ become too small.  

\section{Proof of Ryu-Takayanagi formula for pseudoentangled holographic tensor network \label{sec:si_rt}}

Here we prove a lower bound on the average entropy $S(A)$ of a subsystem $A$ in the pseudorandom holographic tensor network ensemble, $\holo$, considered in the main text. 
Our strategy for lower-bounding $S(A)$ is to evaluate the average purity, $\mathbb{E}[{\rm Tr}(\rho_A^2)]$. Indeed we have the following:
\begin{fact} 
 We have, in general, $-\log\mathbb{E}[{\rm Tr}(\rho_A^2)] \leq \mathbb{E}[S(A)]$.
\end{fact}
\begin{proof}
First, by convexity, we have $\mathbb{E}[-\log(x)] \geq -\log\mathbb{E}[x]$, and thus
\begin{equation}
-\log\mathbb{E}[{\rm Tr}(\rho_A^2)] \leq \mathbb{E}[-\log {\rm Tr}(\rho_A^2)] = \mathbb{E}[S_2(A)],
\end{equation}
where we identified the second Renyi entropy $S_2(A) = -\log {\rm Tr}(\rho_A^2)$. 
Secondly, the well known monotonicity of the Renyi entropies $S_n(A)$ vs the index $n$, and the fact that $S(A) = \lim_{n\to 1} S_n(A)$, gives $S_2(A) \leq S(A)$ in general, and thus $\mathbb{E}[S(A)] \geq \mathbb{E}[S_2(A)]$.
\end{proof}

In the following we treat the unitary gates as if they are genuinely Haar-random: since the average purity is a second-moment quantity, this is justified by noting that the PFC ensemble forms an exact two-design (due to the random Clifford operation that is part of the PFC sequence). 
The average purity for our tensor network states is given by a partition function of spins valued in $S_2 = \{e,s\}$:
\begin{equation}
    {\rm Tr}(\rho_A^2) = \frac{1}{D} \sum_{\{\sigma_i,\tau_i \in S_2 \}} e^{-E_{A,\bar{A}}[\{\sigma_i,\tau_i\}]},
\end{equation}
where the energy term $E_{A,\bar{A}}[\{\sigma_i,\tau_i\}]$ comprises both bulk terms (overlaps or  Weingarten functions between neighboring permutations) and boundary terms (overlaps between $\sigma$'s and $e$ in $A$, and between $\sigma$'s and $s$ in $\bar{A}$, respectively). 
The prefactor of $1/D = \chi^{-L}$ comes from the normalization of the input Bell pairs, $\frac{1}{\sqrt \chi} \sum_{j=1}^\chi \ket{jj}$, contributing a factor of $\chi^{-2}$ to $\rho_A^2$ for each pair of output qudits, hence overall $(\chi^{-2})^{L/2} = \chi^{-L} = 1/D$. 

Since we are interested in an upper bound on the purity (i.e. a lower bound on the entropy), we can switch to positive Boltzmann weights: 
\begin{equation}
    {\rm Tr}(\rho_A^2) 
    \leq \frac{1}{D} \sum_{\{\sigma_i,\tau_i \in S_2 \}} \left| e^{- E_{A,\bar{A}}[\{\sigma_i,\tau_i\}]}\right| 
    = \frac{1}{D} \sum_{\{\sigma_i,\tau_i \in S_2 \}} e^{-E^+_{A,\bar{A}}[\{\sigma_i,\tau_i\}]},
\end{equation}
The graph contains two types of bonds: overlaps between permutations (solid lines in Fig.~2(d)) and Weingarten functions (dashed lines in Fig.~2(d)). 
The overlaps contribute a factor of either $\chi^2$ (if the two permutations they connect are the same) or $\chi$ (if they are opposite); the Weingarten functions are either $(\chi^6-1)^{-1}$ or $\chi^{-3} (\chi^6-1)^{-1}$ respectively. 
Let us define $n_o$ and $n_w$ as the number of overlap bonds and Weingarten bonds, respectively; by counting the number of edges in the graph, we have $n_o = 3 n_w + L/2$ (recall $L$ is the number of output qudits).
Factoring out a $\chi^2$ from each of the $n_o$ overlap bonds and $(\chi^6-1)^{-1}$ from each of the $n_w$ Weingarten bonds, we get a factor of $\chi^{2n_o} / (\chi^6-1)^{n_w} = \chi^L (1-\chi^{-6})^{-n_w}$, thus  
\begin{align}
    {\rm Tr}(\rho_A^2) 
    & \leq ( 1-\chi^{-6})^{-n_w} \sum_{\{\sigma_i,\tau_i \in S_2 \}} e^{-\delta E^+_{A,\bar{A}}[\{\sigma_i,\tau_i\}]}
    \leq [1+O(N/\chi^6)] \sum_{\{\sigma_i,\tau_i \in S_2 \}} e^{-\delta E^+_{A,\bar{A}}[\{\sigma_i,\tau_i\}]}.
\end{align}
The energy terms are now normalized in such a way that they contribute a factor of 1 if the two permutations are equal, and $1/\chi$ (for overlap bonds) or $1/\chi^3$ (for Weingarten bonds) respectively if the permutations are different.
Since the partition function depends only on the relative value of neighboring spins, the pinned values of the spins at the boundary ($s$ in $A$ and $e$ in $\bar{A}$) allow us to turn the sum over site variables $\{\sigma_i,\tau_i\}$ into a sum over bond variables. This in turn can be rewritten in terms of cuts $\gamma$ through the bonds of the network (a bond is `cut' if and only if its spin variable is the transposition $s\in S_2$), giving
\begin{equation}
    {\rm Tr}(\rho_A^2) 
    \leq [1+O(N/\chi^6)]\sum_{\text{cuts }\gamma} \chi^{-|\gamma|},
\end{equation}
where $|\gamma|$ is the number of bonds intersected by $\gamma$ (with the convention that a Weingarten bond counts as 3 bonds in this case). This is a RT formula with a fluctuating cut; then, owing to the large bond dimension $\chi$, fluctuations of the cut are suppressed and one obtains `mininum-cut' prescription. 

To show this, let us (very loosely) upper bound the number of possible cuts of a given length $\ell$ as $\binom{n_b}{\ell}$, with $n_b = n_w + n_o$ the total number of bonds in the network (on hyperbolic tilings, where bulk and boundary scale in the same way, we have $n_b = O(N)$). 
Then, letting $\ell_{\rm min} = \min_{\gamma_A} |\gamma_A|$, we have 
\begin{align}
    {\rm Tr}(\rho_A^2) 
    & \leq [1+O(N/\chi^6)] \sum_{\ell = \ell_{\rm min}}^{n_b} \binom{n_b}{\ell} \chi^{-\ell}
    \leq [1+O(N/\chi^6)] \sum_{\ell = \ell_{\rm min}}^{\infty} (n_b / \chi)^{\ell} \\
    & \leq  \frac{1+O(N/\chi^6)}{1-(n_b / \chi)} (n_b / \chi)^{\ell_{\rm min}}
\end{align}

Finally, we arrive at the desired bound: 
\begin{align}
\mathbb{E}[S(A)] 
& \geq -\log \mathbb{E}[{\rm Tr}(\rho_A^2)] 
\geq \ell_{\rm min} \log \frac{\chi}{n_b} + \log\left(1+O(N/\chi) \right) \\
& \geq \ell_{\rm min} \log(\chi)\left[1-\frac{\log n_b}{\log \chi} \right] + O(N/\chi)
\end{align}
Since $n_b = O(N)$ and $\log\chi = \omega(\log N)$, we have $\log(n_b) / \log(\chi) = o(1)$ (vanishing in the thermodynamic limit). This gives the result quoted in the main text. 

Nothing in this proof specifically hinges on the choice of holographic network; e.g., all the steps would work analogously for arbitrary tilings of the hyperbolic plane that admit an isometric realization (e.g. this includes all tensors with an even number of legs, and potentially other cases with suitable choices of bond dimensions). 

\end{document}